\documentclass[11pt]{article}
\usepackage{fullpage}
\usepackage{hyperref}
\usepackage{palatino}  
\usepackage{mathpazo}
\usepackage{amssymb,amsmath,amsthm}
\usepackage{mathtools}

\DeclarePairedDelimiter\bracks{[}{]}
\DeclarePairedDelimiter\abs{\lvert}{\rvert}
\DeclarePairedDelimiter\parens{(}{)}

\newtheorem{theorem}{Theorem}[section]

\newtheorem{definition}[theorem]{Definition}

\newtheorem{claim}[theorem]{Claim}
\newtheorem{lemma}[theorem]{Lemma}

\newtheorem{corollary}[theorem]{Corollary}

\newcommand{\qedsymb}{\hfill{\rule{2mm}{2mm}}}
\renewenvironment{proof}[1][]{\begin{trivlist}
\item[\hspace{\labelsep}{\bf\noindent Proof#1:\/}] }{\qedsymb\end{trivlist}}

\def\calG{{\cal G}}
\def\calD{{\cal D}}

\def\calH{{\cal H}}

\def\C{\mathbb{C}}

\newcommand\expectation{\mathop{{\mathbb{E}}}}

\newcommand{\Real}{\mathsf{Re}}
\newcommand{\Image}{\mathsf{Im}}

\newcommand{\eps}{\epsilon}
\renewcommand{\epsilon}{\varepsilon}

\begin{document}

\title{{\bf The Restricted Isometry Property of Subsampled Fourier Matrices}}

\author{
Ishay Haviv\thanks{School of Computer Science, The Academic College of Tel Aviv-Yaffo, Tel Aviv 61083, Israel.}
\and Oded Regev\thanks{Courant Institute of Mathematical Sciences, New York University. Supported by the Simons Collaboration on Algorithms and Geometry and by the National Science Foundation (NSF) under Grant No.~CCF-1320188. Any opinions, findings, and conclusions or recommendations expressed in this material are those of the authors and do not necessarily reflect the views of the NSF.}
}

\date{}

\maketitle

\begin{abstract}
A matrix $A \in \C^{q \times N}$ satisfies the {\em restricted isometry property} of order $k$ with constant $\eps$ if it preserves the $\ell_2$ norm of all $k$-sparse vectors up to a factor of $1\pm \eps$.
We prove that a matrix $A$ obtained by randomly
sampling $q = O(k \cdot \log^2 k \cdot \log N)$ rows
from an $N \times N$ Fourier matrix
satisfies the restricted isometry property of order $k$ with a fixed $\eps$ with high probability.
This improves on Rudelson and Vershynin (Comm. Pure Appl. Math.,~2008),
its subsequent improvements, and Bourgain (GAFA Seminar Notes,~2014).
\end{abstract}


\section{Introduction}

A matrix $A \in \C^{q \times N}$ satisfies the {\em restricted isometry property} of order $k$ with constant $\eps>0$ if for every $k$-sparse vector $x \in \C^N$ (i.e., a vector with at most $k$ nonzero entries), it holds that
\begin{equation}\label{eq:ripdef}
 (1-\eps) \cdot \|x\|_2^2 \leq \|Ax\|_2^2 \leq (1+\eps) \cdot \|x\|_2^2 \; .
\end{equation}
Intuitively, this means that every $k$ columns of $A$ are nearly orthogonal. This notion, due to Cand{\`e}s and Tao~\cite{CandesT05}, was intensively studied during the last decade and found various applications and connections to several areas of theoretical computer science, including sparse recovery~\cite{Candes08,IndykR13,NelsonPW14}, coding theory~\cite{CheraghchiGV13}, norm embeddings~\cite{BaraniukDDW08,KrahmerW11}, and computational complexity~\cite{BandeiraDMS13,TillmannP14,NatarajanW14}.

The original motivation for the restricted isometry property comes from the area of compressed sensing. There, one
wishes to compress a high-dimensional sparse vector $x \in \C^N$ to a vector $Ax$, where $A \in \C^{q \times N}$
is a measurement matrix that enables reconstruction of $x$ from $Ax$. Typical goals in this context include
minimizing the number of measurements $q$ and the running time of the reconstruction algorithm.
It is known that the restricted isometry property of $A$, for $\eps < \sqrt{2}-1$, is a sufficient condition for
reconstruction. In fact, it was shown in~\cite{CandesT05,CandesRTV05,CandesRT06,Candes08} that under
this condition, reconstruction is equivalent to finding the vector of least $\ell_1$ norm among all vectors that
agree with the given measurements, a task that can be formulated as a linear program~\cite{ChenDS98,DonohoET06},
and thus can be solved efficiently.

The above application leads to the challenge of finding matrices $A \in \C^{q \times N}$ that satisfy the restricted isometry property and have a small number of rows $q$ as a function of $N$ and $k$. (For simplicity, we ignore for now the dependence on $\eps$.) A general lower bound of $q = \Omega(k \cdot \log (N/k))$ is known to follow from~\cite{GarnaevG} (see also~\cite{FoucartPRU10}). Fortunately, there are matrices that match this lower bound, e.g., random matrices whose entries are chosen independently according to the normal distribution~\cite{CandesT}. However, in many applications the measurement matrix cannot be chosen arbitrarily but is instead given by a random sample of rows from a unitary matrix, typically the discrete Fourier transform. This includes, for instance, various tests and experiments in medicine and biology (e.g., MRI~\cite{MRI} and ultrasound imaging~\cite{Ultrasound}) and applications in astronomy (e.g., radio telescopes~\cite{WengerDSGM10}). An advantage of subsampled Fourier matrices is that they support fast matrix-vector multiplication, and as such, are useful for efficient compression as well as for efficient reconstruction based on iterative methods (see, e.g.,~\cite{NeedellT10}).

In recent years, with motivation from both theory and practice, an intensive line of research has aimed to study the restricted isometry property of random sub-matrices of unitary matrices. Letting $A \in \C^{q \times N}$ be a (normalized) matrix whose rows are chosen uniformly and independently from the rows of a unitary matrix $M \in \C^{N \times N}$, the goal is to prove an upper bound on $q$ for which $A$ is guaranteed to satisfy the restricted isometry property with high probability. Note that the fact that the entries of every row of $A$ are not independent makes this question much more difficult than in the case of random matrices with independent entries.

The first upper bound on the number of rows of a subsampled Fourier matrix that satisfies the restricted isometry property was $O(k \cdot \log^6 N)$, which was proved by Cand{\`e}s and Tao~\cite{CandesT}. This was then improved by Rudelson and Vershynin~\cite{RudelsonV08} to $O(k \cdot \log^2 k \cdot \log(k \log N) \cdot \log N)$ (see also~\cite{Rauhut10,Dirksen15} for a simplified analysis with better success probability). A modification of their analysis led to an improved bound of $O(k \cdot \log^3 k \cdot \log N)$ by Cheraghchi, Guruswami, and Velingker~\cite{CheraghchiGV13}, who related the problem to a question on the list-decoding rate of random linear codes over finite fields. Interestingly, replacing the $\log(k \log N)$ term in the bound of~\cite{RudelsonV08} by $\log k$ was crucial for their application.\footnote{Note that the list-decoding result of~\cite{CheraghchiGV13} was later improved by Wootters~\cite{Wootters13} using different techniques.} Recently, Bourgain~\cite{Bourgain14} proved a bound of $O(k \cdot \log k \cdot \log^2 N)$, which is incomparable to those of~\cite{RudelsonV08,CheraghchiGV13} (and has a worse dependence on $\eps$; see below).
We finally mention that the best known lower bound on the number of rows is $\Omega(k \cdot \log N)$~\cite{BandeiraLM15}.

\subsection{Our Contribution}
In this work, we improve the previous bounds and prove the following.
\begin{theorem}[Simplified]\label{thm:MainIntro}
Let $M \in \C^{N \times N}$ be a unitary matrix with entries of absolute value $O(1/\sqrt{N})$, and let $\eps>0$ be a fixed constant. For some $q = O(k \cdot \log^2 k \cdot \log N )$, let $A \in \C^{q \times N}$ be a matrix whose $q$ rows are chosen uniformly and independently from the rows of $M$, multiplied by $\sqrt{N/q}$. Then, with high probability, the matrix $A$ satisfies the restricted isometry property of order $k$ with constant $\eps$.
\end{theorem}

The main idea in our proof is described in Section~\ref{sec:overview}.
We arrived at the proof from our recent work on list-decoding~\cite{HavivR15},
where a baby version of the idea
was used to bound the sample complexity of learning the class of Fourier-sparse
Boolean functions.\footnote{%
The result in~\cite{HavivR15} is weaker in two main respects.
First, it is restricted to the case that $Ax$ is in $\{0,1\}^q$.
This significantly simplifies the analysis and leads to a better bound on the
number of rows of $A$.
Second, the order of quantifiers is switched, namely it shows that
for any sparse $x$, a random subsampled $A$ works with high probability,
whereas for the restricted isometry property we need to show that a random $A$ works for all sparse $x$.}
Like all previous work on this question, our proof can be seen as
a careful union bound applied to a sequence of progressively finer nets,
a technique sometimes known as chaining.
However, unlike the work of Rudelson and Vershynin~\cite{RudelsonV08} and
its improvements~\cite{CheraghchiGV13,Dirksen15}, we avoid the
use of Gaussian processes, the ``symmetrization process,'' and
Dudley's inequality. Instead, and more in line with
Bourgain's proof~\cite{Bourgain14}, we apply the chaining argument
directly to the problem at hand using only elementary arguments.
It would be interesting to see if our proof can be cast in the Gaussian framework of
Rudelson and Vershynin.

We remark that the bounds obtained in the previous works~\cite{RudelsonV08,CheraghchiGV13} have a multiplicative $O(\eps^{-2})$ term, where a much worse term of $O(\eps^{-6})$ was obtained in~\cite{Bourgain14}. In our proof of Theorem~\ref{thm:MainIntro} we nearly obtain the best known dependence on $\eps$. For simplicity of presentation we first prove in Section~\ref{sec:simple} our bound with a weaker multiplicative term of $O(\eps^{-4})$, and then, in Section~\ref{sec:improved}, we modify the analysis and decrease the dependence on $\eps$ to $\widetilde{O}(\eps^{-2})$.

\subsection{Related Literature}\label{sec:related}

As mentioned before, one important advantage of using subsampled Fourier matrices in compressed sensing is that they support fast, in fact nearly linear time, matrix-vector multiplication.
In certain scenarios, however, one is not restricted to using subsampled Fourier matrices
as the measurement matrix. The question then is whether
one can decrease the number of rows using another measurement matrix,
while still keeping the near-linear multiplication time.
For $k < N^{1/2-\gamma}$ where $\gamma>0$ is an arbitrary constant, the
answer is yes: a construction with the \emph{optimal} number $O(k \cdot \log N)$
of rows follows from works by Ailon and Chazelle~\cite{AilonC09} and Ailon and Liberty~\cite{AilonL09} (see~\cite{BaraniukDDW08}).
For general $k$, Nelson, Price, and Wootters~\cite{NelsonPW14} suggested
taking subsampled Fourier matrices and ``tweaking'' them by bunching
together rows with random signs.
Using the Gaussian-process-based analysis of~\cite{RudelsonV08,CheraghchiGV13}
and introducing further techniques from~\cite{KrahmerMR12},
they showed that with this construction one can reduce the number of rows by a logarithmic
factor to $O(k \cdot \log^2 (k \log N) \cdot \log N)$ while still keeping the nearly linear multiplication time.
Our result shows that the same number of rows (in fact, a slightly smaller number)
can be achieved already with the original subsampled Fourier matrices without having
to use the ``tweak.''
A natural open question is whether the ``tweak'' from~\cite{NelsonPW14} and their techniques
can be combined with ours to further reduce the number of rows.
An improvement in the regime of parameters of $k = \omega(\sqrt{N})$ would lead to more efficient low-dimensional embeddings based on Johnson--Lindenstrauss matrices (see, e.g.,~\cite{AilonC09,AilonL09,KrahmerW11,AilonL13,NelsonPW14}).

\subsection{Proof Overview}\label{sec:overview}

Recall from Theorem~\ref{thm:MainIntro} and from~\eqref{eq:ripdef} that our goal
is to prove that a matrix $A$ given by a random sample  $Q$ of $q$ rows of $M$ satisfies
with high probability that for all $k$-sparse $x$, $\|Ax\|_2^2 \approx \|x\|_2^2$.
Since $M$ is unitary, the latter is equivalent to saying that
$
\|Ax\|_2^2 \approx \|Mx\|_2^2
$. Yet another way of expressing this condition is as
\[
\expectation_{j \in Q} \bracks[\big]{(|Mx|^2)_j} \approx \expectation_{j \in [N]}\bracks[\big]{(|Mx|^2)_j} \; ,
\]
i.e., that a sample $Q \subseteq [N]$ of $q$ coordinates of the vector $|Mx|^2$ gives a
good approximation to the average of all its coordinates. Here, $|Mx|^2$ refers
to the vector obtained by taking the squared absolute value of $Mx$ coordinate-wise.
For reasons that will become clear soon, it will be convenient to assume
without loss of generality that $\|x\|_1=1$. With this scaling, the sparsity assumption
implies
that $\|Mx\|_2^2$ is not too small (namely at least $1/k$), and this
will determine the amount of additive error we can afford in the approximation above.
This is the only way we use the sparsity assumption.

At a high level, the proof proceeds by defining a finite set of vectors $\calH$
that forms a \emph{net}, i.e., a set satisfying that any vector $|Mx|^2$ is
close to one of the vectors in $\calH$. We then argue using the
Chernoff-Hoeffding bound that for any fixed vector
$h \in \calH$, a sample of $q$ coordinates gives a good approximation to the
average of $h$. Finally, we complete the proof by a union bound over all $h \in \calH$.

In order to define the set $\calH$ we notice that since $\|x\|_1=1$,
$Mx$ can be seen as a weighted average of
the columns of $M$ (possibly with signs).
In other words, we can think of $Mx$ as the
\emph{expectation} of a vector-valued random variable given by a certain
probability distribution over the columns of $M$.
Using the Chernoff-Hoeffding bound again,
this implies that we can approximate $Mx$ well by taking the average
over a small number of samples from this distribution.
We then let $\calH$ be the set of all possible such averages,
and a bound on the cardinality of $\calH$ follows easily
(basically $N$ raised to the number of samples).
This technique is sometimes referred to as Maurey's empirical method.

The argument above is actually oversimplified, and carrying it out
leads to rather bad bounds on $q$. As a result, our proof in
Section~\ref{sec:simple} is slightly more delicate.
Namely, instead of just one set $\calH$, we have
a sequence of sets, $\calH_1,\calH_2, \ldots$,
each being responsible for approximating a different
scale of $|Mx|^2$. The first set $\calH_1$ approximates
$|Mx|^2$ on coordinates on which its value is highest;
since the value is high, we need less samples
in order to approximate it well,  as a result of which
the set $\calH_1$ is small. The next set $\calH_2$
approximates $|Mx|^2$ on coordinates on which its value is
somewhat smaller, and is therefore a bigger set, and so on and so forth.
The end result is that any vector $|Mx|^2$ can be
approximately decomposed into a sum $\sum_i h^{(i)}$,
with $h^{(i)} \in \calH_i$. To complete the proof, we argue
that a random choice of $q$ coordinates approximates all
the vectors in all the $\calH_i$ well.
The reason working with several $\calH_i$ leads to
the better bound stated in Theorem~\ref{thm:MainIntro}
is this: even though as $i$ increases the number of vectors
in $\calH_i$ grows, the quality of approximation that we
need the $q$ coordinates to provide decreases, since
the value of $|Mx|^2$ there is small and so errors are less
significant. It turns out that these two requirements on $q$
balance each other perfectly, leading to the desired bound on $q$.

\paragraph{Acknowledgments.}
We thank Afonso S.~Bandeira, Mahdi Cheraghchi, Michael Kapralov, Jelani Nelson, and Eric Price for useful discussions, and anonymous reviewers for useful comments.

\section{Preliminaries}

\paragraph{Notation.}
The notation $x \approx_{\eps,\alpha} y$ means that $x \in [(1-\eps)y-\alpha,(1+\eps)y+\alpha]$.
For a matrix $M$, we denote by $M^{(\ell)}$ the $\ell$th column of $M$ and define $\|M\|_\infty = \max_{i,j}{|M_{i,j}|}$.

\paragraph{The Restricted Isometry Property.}
The restricted isometry property is defined as follows.
\begin{definition}\label{def:RIP}
We say that a matrix $A \in \C^{q \times N}$ satisfies the {\em restricted isometry property} of order $k$ with constant $\eps$ if for every $k$-sparse vector $x \in \C^N$ it holds that
\[(1-\eps) \cdot \|x\|_2^2 \leq \|Ax\|_2^2 \leq (1+\eps) \cdot \|x\|_2^2.\]
\end{definition}

\paragraph{Chernoff-Hoeffding Bounds.}
We now state the Chernoff-Hoeffding bound
(see, e.g.,~\cite{McDiarmid98}) and derive several simple corollaries
that will be used extensively later.

\begin{theorem}\label{thm:Chernoff}
Let $X_1,\ldots,X_N$ be $N$ identically distributed independent random variables in $[0,a]$ satisfying $\expectation[X_i]=\mu$ for all $i$, and denote $\overline{X} = \frac{1}{N} \cdot \sum_{i=1}^{N}{X_i}$. Then there exists a universal constant $C$ such that for every $0< \eps \leq 1/2$, the probability that $\overline{X} \approx_{\eps,0} \mu$ is at least $1-2e^{-C \cdot N \mu \eps^2/a}$.
\end{theorem}

\begin{corollary}\label{cor:ChernoffCombined}
Let $X_1,\ldots,X_N$ be $N$ identically distributed independent random variables in $[0,a]$ satisfying $\expectation[X_i]=\mu$ for all $i$, and denote $\overline{X} = \frac{1}{N} \cdot \sum_{i=1}^{N}{X_i}$. Then there exists a universal constant $C$ such that for every $0< \eps \leq 1/2$ and $\alpha >0$,
the probability that $\overline{X} \approx_{\eps,\alpha} \mu$ is at least $1-2e^{-C \cdot N \alpha \eps/a}$.
\end{corollary}

\begin{proof}
If $\mu \geq \frac{\alpha}{\eps}$ then by Theorem~\ref{thm:Chernoff} the probability that $\overline{X} \approx_{\eps,0} \mu$ is at least $1-2e^{-C \cdot N \mu \eps^2/a}$, which is at least $1-2e^{-C \cdot N \alpha \eps/a}$. Otherwise, Theorem~\ref{thm:Chernoff} for $\tilde{\eps} = \frac{\alpha}{\mu} > \eps$ implies that the probability that $\overline{X} \approx_{\tilde{\eps},0} \mu$, hence $\overline{X} \approx_{0,\alpha} \mu$, is at least $1-2e^{-C \cdot N \mu \tilde{\eps}^2/a}$, and the latter is at least $1-2e^{-C \cdot N \alpha \eps/a}$.
\end{proof}

\begin{corollary}\label{cor:ChernoffCombinedNew}
Let $X_1,\ldots,X_N$ be $N$ identically distributed independent random variables in $[-a,+a]$ satisfying $\expectation[X_i]=\mu$ and $\expectation[|X_i|]=\tilde{\mu}$ for all $i$, and denote $\overline{X} = \frac{1}{N} \cdot \sum_{i=1}^{N}{X_i}$. Then there exists a universal constant $C$ such that for every $0< \eps' \leq 1/2$ and $\alpha >0$,
the probability that $\overline{X} \approx_{0,\eps' \cdot \tilde{\mu} + \alpha} \mu$ is at least $1-4e^{-C \cdot N \alpha \eps'/a}$.
\end{corollary}

\begin{proof}
The corollary follows by applying Corollary~\ref{cor:ChernoffCombined} to $\max (X_i, 0)$ and to $-\min (X_i, 0)$.
\end{proof}

We end with the additive form of the bound, followed by an easy extension to the complex case.

\begin{corollary}\label{thm:Chernoff1}
Let $X_1,\ldots,X_N$ be $N$ identically distributed independent random variables in $[-a,+a]$ satisfying $\expectation[X_i]=\mu$ for all $i$, and denote $\overline{X} = \frac{1}{N} \cdot \sum_{i=1}^{N}{X_i}$. Then there exists a universal constant $C$ such that for every $b>0$,
the probability that $\overline{X} \approx_{0,b} \mu$ is at least $1-4e^{-C \cdot N b^2/a^2}$.
\end{corollary}

\begin{proof}
We can assume that $b \leq 2a$. The corollary follows by applying Corollary~\ref{cor:ChernoffCombinedNew} to, say, $\alpha = 3b/4$ and $\eps' = b/(4a)$.
\end{proof}

\begin{corollary}\label{cor:ChernoffComplex}
Let $X_1,\ldots,X_N$ be $N$ identically distributed independent complex-valued random variables satisfying $|X_i| \leq a$ and $\expectation[X_i]=\mu$ for all $i$, and denote $\overline{X} = \frac{1}{N} \cdot \sum_{i=1}^{N}{X_i}$. Then there exists a universal constant $C$ such that for every $b>0$, the probability that $|\overline{X}| \approx_{0,b} |\mu|$ is at least $1-8e^{-C \cdot N b^2/a^2}$.
\end{corollary}

\begin{proof}
By Corollary~\ref{thm:Chernoff1} applied to the real and imaginary parts of the random variables $X_1,\ldots,X_N$ it follows that for a universal constant $C$, the probability that $\Real(\overline{X}) \approx_{0,b/\sqrt{2}} \Real(\mu)$ and $\Image(\overline{X}) \approx_{0,b/\sqrt{2}} \Image(\mu)$ is at least $1-8e^{-C \cdot N b^2/a^2}$. By triangle inequality, it follows that with such probability we have $|\overline{X}| \approx_{0,b} |\mu|$, as required.
\end{proof}

\section{The Simpler Analysis}\label{sec:simple}

In this section we prove our result with a multiplicative term of $O(\eps^{-4})$ in the bound. We start with the following theorem.

\begin{theorem}\label{thm:Main_s}
For a sufficiently large $N$, a matrix $M \in \C^{N \times N}$, and sufficiently small $\eps, \eta > 0$, the following holds.
For some
$q = O  (\eps^{-3} \eta^{-1} \log N \cdot \log^2( 1/\eta)  )$,
let $Q$ be a multiset of $q$ uniform and independent random elements of $[N]$.
Then, with probability $1- 2^{-\Omega(\eps^{-2} \cdot \log N \cdot \log (1/\eta))}$, it holds that for every $x \in \C^N$,
\[\expectation_{j \in Q} \bracks[\big]{|(Mx)_j|^2} \approx_{\eps, \eta \cdot \|x\|_1^2 \cdot \|M\|_\infty^2} \expectation_{j \in [N]}\bracks[\big]{|(Mx)_j|^2}.\]
\end{theorem}

Throughout the proof we assume without loss of generality that the matrix $M \in \C^{N \times N}$ satisfies $\|M\|_\infty = 1$.
For $\eps, \eta >0$, we denote $t = \log_2(1/\eta)$, $r = \log_2(1/\eps^2)$, and $\gamma = \eta/(2t)$. We start by defining several vector sets as follows.

\paragraph{The Vector Sets $\calG_i$.}
For every $1 \leq i \leq t+r$, let $\calG_i$ denote the set of all vectors $g^{(i)} \in \C^N$ that can be represented as
\begin{equation}\label{eq:vectorg}
g^{(i)} = \frac{\sqrt{2}}{|F|} \cdot \sum_{(\ell,s) \in F}{(-1)^{s/2} \cdot M^{(\ell)}}
\end{equation}
for a multiset $F$ of $O(2^i \cdot \log(1/\gamma))$ pairs in $[N] \times \{0,1,2,3\}$.
A trivial counting argument gives the following.

\begin{claim}\label{claim:sizeGi}
For every $1 \leq i \leq t+r$, $|\calG_i| \leq N^{O (2^i \cdot \log (1/\gamma))}.$
\end{claim}

\paragraph{The Vector Sets $\calH_i$.}
For a $t$-tuple of vectors $(g^{(1+r)},\ldots,g^{(t+r)}) \in \calG_{1+r} \times \cdots \times \calG_{t+r}$ and for $1 \leq i \leq t$, let $B_i$ be the set of all $j \in [N]$ for which $i$ is the smallest index satisfying $|g^{(i+r)}_j| \geq 2 \cdot 2^{-i/2}$. For such $i$, define the vector $h^{(i)}$ by
\begin{align}\label{eq:H_i}
h^{(i)}_j = \min (|g^{(i+r)}_j|^2 \cdot \mathbb{1}_{j \in B_i}, 9 \cdot 2^{-i}).
\end{align}
Let $\calH_i$ be the set of all vectors $h^{(i)}$ that can be obtained in this way.

\begin{claim}\label{claim:sizeHi}
For every $1 \leq i \leq t$, $|\calH_i| \leq N^{O (\eps^{-2} \cdot 2^i \cdot \log (1/\gamma))}.$
\end{claim}

\begin{proof}
Observe that every $h^{(i)} \in \calH_i$ is fully defined by some $(g^{(1+r)},\ldots,g^{(i+r)}) \in \calG_{1+r} \times \cdots \times \calG_{i+r}$. Hence
\begin{align*}
|\calH_i| \leq
|\calG_{1+r}| \cdots |\calG_{i+r}|
 \leq
N^{O (\log (1/\gamma)) \cdot (2^{1+r}+2^{2+r}+\cdots +2^{i+r})}
 \leq
N^{O (\log (1/\gamma)) \cdot 2^{i+r+1}} \; .
\end{align*}
Using the definition of $r$, the claim follows.
\end{proof}

\begin{lemma}\label{lemma:QisGood}
For every $\tilde{\eta} > 0$ and some $q = O(\eps^{-3} \tilde{\eta}^{-1} \log N \cdot \log(1/\gamma))$, let $Q$ be a multiset of $q$ uniform and independent random elements of $[N]$.
Then, with probability $1- 2^{-\Omega(\eps^{-2} \cdot \log N \cdot \log (1/\gamma))}$, it holds that for all $1 \leq i \leq t$ and $h^{(i)} \in \calH_i$ ,
\[\expectation_{j \in Q} \bracks[\big]{h^{(i)}_j} \approx_{\eps, \tilde{\eta}} \expectation_{j \in [N]}\bracks[\big]{h^{(i)}_j}.\]
\end{lemma}

\begin{proof}
Fix an $1 \leq i \leq t$ and a vector $h^{(i)} \in \calH_i$, and denote
$\mu = \expectation_{j \in [N]}[h^{(i)}_j]$. By Corollary~\ref{cor:ChernoffCombined},
applied with $\alpha = \tilde{\eta}$ and $a = 9 \cdot 2^{-i}$ (recall that $h^{(i)}_j \leq a$ for every $j$), with probability $1-2^{-\Omega(2^i \cdot q\eps \tilde{\eta})}$, it holds that $\expectation_{j \in Q}[h^{(i)}_j] \approx_{\eps,\tilde{\eta}} \mu$.
Using Claim~\ref{claim:sizeHi}, the union bound over all
the vectors in $\calH_i$ implies that the probability
that some $h^{(i)} \in \calH_i$ does not satisfy $\expectation_{j \in Q}[h^{(i)}_j] \approx_{\eps,\tilde{\eta}} \mu$ is at
most
\[
N^{O (\eps^{-2} \cdot 2^i \cdot \log (1/\gamma))} \cdot 2^{-\Omega (2^i \cdot q\eps \tilde{\eta})} \leq 2^{-\Omega(\eps^{-2} \cdot 2^i \cdot \log N \cdot \log (1/\gamma))} \; .
\]
We complete the proof by a union bound over $i$.
\end{proof}

\paragraph{Approximating the Vectors $Mx$.}

\begin{lemma}\label{lemma:approx}
For every vector $x \in \C^N$ with $\|x\|_1 = 1$,
every multiset $Q \subseteq [N]$, and every $1 \leq i \leq t+r$, there exists a vector $g \in \calG_i$ that
satisfies $|(Mx)_j| \approx_{0, 2^{-i/2}} |g_j|$
for all but at most $\gamma$ fraction of $j \in [N]$ and for all but at most $\gamma$ fraction of $j \in Q$.
\end{lemma}

\begin{proof}
Observe that for every $\ell \in [N]$ there exist
$p_{\ell,0},p_{\ell,1},p_{\ell,2},p_{\ell,3} \ge 0$ that satisfy
\[ \sum_{s=0}^{3}{p_{\ell,s}} = |x_\ell| \mbox{ ~~~~and~~~~ }\sqrt{2} \cdot \sum_{s=0}^{3}{p_{\ell,s} \cdot (-1)^{s/2}} = x_\ell.\]
Notice that the assumption $\|x\|_1 = 1$ implies that the numbers $p_{\ell,s}$
form a probability distribution. Thus, the vector $Mx$ can be represented as
\[Mx = \sum_{\ell=1}^{N}{x_\ell \cdot M^{(\ell)}} = \sqrt{2} \cdot  \sum_{\ell=1}^{N} {\sum_{s=0}^{3}{p_{\ell,s} \cdot (-1)^{s/2} \cdot M^{(\ell)}}} = \expectation_{(\ell,s) \sim D} [\sqrt{2} \cdot (-1)^{s/2} \cdot M^{(\ell)} ],\]
where $D$ is the distribution that assigns probability $p_{\ell,s}$ to the pair $(\ell,s)$.

Let $F$ be a multiset of $O(2^i \cdot \log(1/\gamma))$ independent random samples from $D$, and let $g \in \calG_i$ be the vector
corresponding to $F$ as in~\eqref{eq:vectorg}.
By Corollary~\ref{cor:ChernoffComplex}, applied with $a = \sqrt{2}$ (recall that $\|M\|_\infty = 1$) and $b = 2^{-i/2}$, for every $j \in [N]$ the probability that
\begin{align}\label{eq:Mxj}
|(Mx)_j| \approx_{0, 2^{-i/2}} |g_j|
\end{align}
is at least $1-\gamma/4$. It follows that the expected number of $j \in [N]$ that do not satisfy~\eqref{eq:Mxj} is at most $\gamma N /4$, so by Markov's inequality the probability that the number of $j \in [N]$ that do not satisfy~\eqref{eq:Mxj} is at most $\gamma N$ is at least $3/4$. Similarly, the expected number of $j \in Q$ that do not satisfy~\eqref{eq:Mxj} is at most $\gamma |Q| /4$, so by Markov's inequality,
with probability at least $3/4$ it holds that
the number of $j \in Q$ that do not satisfy~\eqref{eq:Mxj} is at most $\gamma |Q|$.
It follows that there exists a vector $g \in \calG_i$ for which~\eqref{eq:Mxj} holds for all but at most $\gamma$ fraction of $j \in [N]$ and for all but at most $\gamma$ fraction of $j \in Q$, as required.
\end{proof}

\begin{lemma}\label{lemma:exists_gi}
For every multiset $Q \subseteq [N]$ and every vector $x \in \C^N$ with $\|x\|_1 = 1$ there exists a $t$-tuple of vectors $(h^{(1)},\ldots,h^{(t)}) \in \calH_1 \times \cdots \times \calH_t$ for which
\begin{enumerate}
  \item $\expectation_{j \in Q} \bracks[\big]{ |(Mx)_j|^2 } \approx_{O(\eps), O(\eta)}
	       \expectation_{j \in Q} \bracks[\big]{ \sum_{i=1}^{t}{h^{(i)}_j} }$
		    and
  \item $\expectation_{j \in [N]} \bracks[\big]{|(Mx)_j|^2 } \approx_{O(\eps), O(\eta)}
	       \expectation_{j \in [N]} \bracks[\big]{\sum_{i=1}^{t}{h^{(i)}_j} }$.
\end{enumerate}
\end{lemma}

\begin{proof}
By Lemma~\ref{lemma:approx}, for every $1 \leq i \leq t$ there exists a vector $g^{(i+r)} \in \calG_{i+r}$ that satisfies
\begin{align}\label{eq:Mxj_2}
|(Mx)_j| \approx_{0, 2^{-(i+r)/2}} |g^{(i+r)}_j|
\end{align}
for all but at most $\gamma$ fraction of $j \in [N]$ and for all but at most $\gamma$ fraction of $j \in Q$.
We say that $j \in [N]$ is {\em good} if~\eqref{eq:Mxj_2} holds for every $1 \leq i \leq t$, and otherwise that it is {\em bad}.
Notice that all but at most $t \gamma$ fraction of $j \in [N]$ are good
and that all but at most $t \gamma$ fraction of $j \in Q$ are good.
Let $(h^{(1)},\ldots,h^{(t)})$ and $(B_1,\ldots,B_t)$ be the vectors and sets associated with $(g^{(1+r)}, \ldots, g^{(t+r)})$ as defined in~\eqref{eq:H_i}.
We claim that $h^{(1)},\ldots,h^{(t)}$ satisfy the requirements of the lemma.


We first show that
for every good $j$ it holds that
$|(Mx)_j|^2 \approx_{3\eps,9\eta} \sum_{i=1}^{t}{h^{(i)}_j}$.
To obtain it, we observe that if $j \in B_i$ for some $i$, then
\begin{align}\label{eq:g_i+r}
2 \cdot 2^{-i/2} \leq |g^{(i+r)}_j| \leq 3 \cdot 2^{-i/2}.
\end{align}
The lower bound follows simply from the definition of $B_i$. For the upper bound, which trivially holds for $i=1$, assume that $i \geq 2$, and notice that the definition of $B_i$ implies that
$|g^{(i+r-1)}_j| < 2 \cdot 2^{-(i-1)/2}$.
Using~\eqref{eq:Mxj_2}, and assuming that $\eps$ is sufficiently small, we obtain that
\begin{align*}
|g^{(i+r)}_j| &\leq |(Mx)_j| + 2^{-(i+r)/2} \leq |g^{(i+r-1)}_j| + 2^{-(i+r-1)/2} + 2^{-(i+r)/2} \\
&\leq  2^{-i/2}(2^{3/2} + 2^{1/2} \cdot \eps+\eps) \leq 3 \cdot 2^{-i/2}.
\end{align*}
Hence, by the upper bound in~\eqref{eq:g_i+r}, for a good $j \in B_i$ we have $h^{(i)}_j = |g^{(i+r)}_j|^2$ and $h^{(i')}_j = 0$ for $i' \neq i$. Observe that by the lower bound in~\eqref{eq:g_i+r},
\[|(Mx)_j| \in  [|g^{(i+r)}_j| - 2^{-(i+r)/2}, |g^{(i+r)}_j| + 2^{-(i+r)/2}] \subseteq [(1-\eps) \cdot |g^{(i+r)}_j| , (1+\eps) \cdot |g^{(i+r)}_j|], \]
and that this implies that $|(Mx)_j|^2 \approx_{3\eps,0} \sum_{i=1}^{t}{h^{(i)}_j}$. On the other hand, in case that $j$ is good but does not belong to any $B_i$, recalling that $t = \log_2(1/\eta)$, it follows that
\[|(Mx)_j| \leq |g^{(t+r)}_j| + 2^{-(t+r)/2} \leq 2 \cdot 2^{-t/2} + 2^{-(t+r)/2} \leq 3 \cdot 2^{-t/2} \leq 3\sqrt{\eta},\]
and thus $|(Mx)_j|^2 \approx_{0,9\eta} 0 = \sum_{i=1}^{t}{h^{(i)}_j}$.

Finally, for every bad $j$ we have
\[
\abs[\Big]{|(Mx)_j|^2 - \sum_{i=1}^{t}{h^{(i)}_j}}
\le
\max \Big ( |(Mx)_j|^2, \sum_{i=1}^{t}{h^{(i)}_j} \Big ) \leq 2.
\]
Since at most $t\gamma$ fraction of the elements in $[N]$ and in $Q$ are bad, their effect on the difference between the expectations in the lemma can be bounded by $2t \gamma$. By our choice of $\gamma$, this is $\eta$, completing the proof of the lemma.
\end{proof}

Finally, we are ready to prove Theorem~\ref{thm:Main_s}.

\begin{proof}[ of Theorem~\ref{thm:Main_s}]
By Lemma~\ref{lemma:QisGood}, applied with $\tilde{\eta} = \eta/(2t)$,
a random multiset $Q$ of size
\begin{align*}
q = O \Big (\eps^{-3} \eta^{-1} \cdot t \cdot \log N \cdot \log(1/\gamma) \Big)
 = O \Big (\eps^{-3} \eta^{-1} \log N \cdot \log^2(1/\eta) \Big)
\end{align*}
satisfies with probability $1- 2^{- \Omega(\eps^{-2} \cdot \log N \cdot \log (1/\eta))}$
that for all $1 \leq i \leq t$ and $h^{(i)} \in \calH_i$,
\begin{align*}
\expectation_{j \in Q} \bracks[\big]{h^{(i)}_j} \approx_{\eps, \eta/t} \expectation_{j \in [N]}\bracks[\big]{h^{(i)}_j} \; ,
\end{align*}
in which case we also have
\[
\expectation_{j \in Q} \bracks[\Big]{\sum_{i=1}^{t}{h^{(i)}_j}} \approx_{\eps, \eta} \expectation_{j \in [N]}\bracks[\Big]{\sum_{i=1}^{t}{h^{(i)}_j}}.
\]

We show that a $Q$ with the above property satisfies the requirement of the theorem. Let $x \in \C^N$ be a vector, and assume without loss of generality that $\|x\|_1 = 1$.
By Lemma~\ref{lemma:exists_gi}, there exists a $t$-tuple of vectors $(h^{(1)},\ldots,h^{(t)}) \in \calH_1 \times \cdots \times \calH_t$
satisfying Items 1 and 2 there.
As a result,
\[
\expectation_{j \in Q} \bracks[\big]{|(Mx)_j|^2}
\approx_{O(\eps), O(\eta)}
\expectation_{j \in [N]}\bracks[\big]{|(Mx)_j|^2} \; ,\]
and we are done.
\end{proof}

\subsection{The Restricted Isometry Property}
Equipped with Theorem~\ref{thm:Main_s}, it is easy to derive our result on the restricted isometry property (see Definition~\ref{def:RIP}) of random sub-matrices of unitary matrices.

\begin{theorem}\label{thm:RIP_simple}
For sufficiently large $N$ and $k$, a unitary matrix $M \in \C^{N \times N}$
satisfying $\|M\|_\infty \leq O(1/\sqrt{N})$, and a sufficiently small
$\eps > 0$, the following holds.
For some $q = O(\eps^{-4}  \cdot k \cdot \log^2(k/\eps) \cdot \log N)$,
let $A \in \C^{q \times N}$ be a matrix whose $q$ rows are chosen
uniformly and independently from the rows of $M$,
multiplied by $\sqrt{N/q}$.
Then, with probability
$1-2^{-\Omega(\eps^{-2} \cdot \log N \cdot \log (k/\eps))}$, the matrix $A$ satisfies the restricted isometry property of order $k$ with constant $\eps$.
\end{theorem}

\begin{proof}
Let $Q$ be a multiset of $q$ uniform and independent random elements of $[N]$,
defining a matrix $A$ as above.
Notice that by the Cauchy-Schwarz inequality,
any $k$-sparse vector $x \in \C^N$ with $\|x\|_2=1$ satisfies $\|x\|_1 \le \sqrt{k}$.
Applying Theorem~\ref{thm:Main_s} with $\eps/2$ and some $\eta = \Omega(\eps/k)$,
we get that with probability $1-2^{-\Omega(\eps^{-2} \cdot \log N \cdot \log (k/\eps))}$,
it holds that for every $x \in \C^N$ with
$\|x\|_2 = 1$,
\[
\|Ax\|_2^2 =
N \cdot \expectation_{j \in Q} \bracks[\big]{|(Mx)_j|^2}
\approx_{\eps/2, \eps/2}
N \cdot \expectation_{j \in [N]}\bracks[\big]{|(Mx)_j|^2}
=
\|Mx\|_2^2 = 1
\; .\]
It follows that every vector $x \in \C^N$ satisfies $\|Ax\|_2^2 \approx_{\eps, 0} \|x\|_2^2$, hence $A$ satisfies the restricted isometry property of order $k$ with constant $\eps$.
\end{proof}

\section{The Improved Analysis}\label{sec:improved}

In this section we prove the following theorem, which improves the bound of Theorem~\ref{thm:Main_s} in terms of the dependence on $\eps$.

\begin{theorem}\label{thm:Main_s_}
For a sufficiently large $N$, a matrix $M \in \C^{N \times N}$, and sufficiently small $\eps, \eta > 0$, the following holds.
For some
$q = O (\log^2(1/\eps) \cdot \eps^{-1} \eta^{-1} \log N \cdot \log^2 (1/\eta) )$,
let $Q$ be a multiset of $q$ uniform and independent random elements of $[N]$.
Then, with probability $1- 2^{-\Omega(\log N \cdot \log (1/\eta))}$, it holds that for every $x \in \C^N$,
\begin{align}\label{eq:thm_improved}
\expectation_{j \in Q} \bracks[\big]{|(Mx)_j|^2} \approx_{\eps, \eta \cdot \|x\|_1^2 \cdot \|M\|_\infty^2} \expectation_{j \in [N]}\bracks[\big]{|(Mx)_j|^2}.
\end{align}
\end{theorem}

We can assume that $\eps \geq \eta$, as otherwise, one can apply the theorem with parameters $\eta/2, \eta/2$ and derive~\eqref{eq:thm_improved} for $\eps,\eta$ as well (because the right-hand size is bounded from above by $\|x\|_1^2 \cdot \|M\|_\infty^2$). As before, we assume without loss of generality that $\|M\|_\infty = 1$. For $\eps \geq \eta >0$, we define $t = \log_2(1/\eta)$ and $r = \log_2(1/\eps^2)$. For the analysis given in this section, we define $\gamma = \eta/(60(t+r))$. Throughout the proof, we use the vector sets $\calG_i$ from Section~\ref{sec:simple} and Lemma~\ref{lemma:approx} for this value of $\gamma$.

\paragraph{The Vector Sets $\calD_{i,m}$.}
For a $(t+r)$-tuple of vectors $(g^{(1)},\ldots,g^{(t+r)}) \in \calG_1 \times \cdots \times \calG_{t+r}$ and for $1 \leq i \leq t$, let $C_i$ be the set of all $j \in [N]$ for which $i$ is the smallest index satisfying $|g^{(i)}_j| \geq 2 \cdot 2^{-i/2}$. For $m = i,\ldots,i+r$ define the vector $h^{(i,m)}$ by
\begin{align}\label{eq:H_i,m}
h^{(i,m)}_j = |g^{(m)}_j|^2 \cdot \mathbb{1}_{j \in C_i},
\end{align}
and for other values of $m$ define $h^{(i,m)} = 0$.
Now, for every $m$, let $\Delta^{(i,m)}$ be the vector defined by
\begin{align}\label{eq:Delta_i,m}
\Delta^{(i,m)}_j = \left\{
  \begin{array}{ll}
    h^{(i,m)}_j-h^{(i,m-1)}_j, & \hbox{if $|h^{(i,m)}_j-h^{(i,m-1)}_j| \leq 30 \cdot 2^{-(i+m)/2}$;} \\
    0, & \hbox{otherwise.}
  \end{array}
\right.
\end{align}
Note that the support of $\Delta^{(i,m)}$ is contained in $C_i$.
Let $\calD_{i,m}$ be the set of all vectors $\Delta^{(i,m)}$ that can be obtained in this way.

\begin{claim}\label{claim:sizeHiNew}
For every $1 \leq i \leq t$ and $i \leq m \leq i+r$, $|\calD_{i,m}| \leq N^{O (2^m \cdot \log (1/\gamma))}.$
\end{claim}

\begin{proof}
Observe that every vector in $\calD_{i,m}$ is fully defined by some $(g^{(1)},\ldots,g^{(m)}) \in \calG_1 \times \cdots \times \calG_m$. Hence
\begin{align*}
|\calD_{i,m}| \leq
|\calG_1| \cdots |\calG_m|
 \leq
N^{O (\log (1/\gamma)) \cdot (2^1+2^2+\cdots+2^m)}
 \leq
N^{O (\log (1/\gamma)) \cdot 2^{m+1}} \; ,
\end{align*}
and the claim follows.
\end{proof}

\begin{lemma}\label{lemma:QisGood_}
For every $\tilde{\eps}, \tilde{\eta} >0$ and some $q = O(\tilde{\eps}^{-1} \tilde{\eta}^{-1} \log N \cdot \log(1/\gamma))$, let $Q$ be a multiset of $q$ uniform and independent random elements of $[N]$. Then, with probability $1- 2^{-\Omega(\log N \cdot \log (1/\gamma))}$, it holds that
for every $1 \leq i \leq t$, $m$, and a vector $\Delta^{(i,m)} \in \calD_{i,m}$ associated with a set $C_i$,
\begin{equation}\label{eq:qapproxdelta}
\expectation_{j \in Q} \bracks[\big]{\Delta^{(i,m)}_j} \approx_{0, b} \expectation_{j \in [N]}\bracks[\big]{\Delta^{(i,m)}_j}
\mbox{ ~~for~~ }
b = O \Big (\tilde{\eps} \cdot 2^{-i} \cdot \frac{|C_i|}{N} + \tilde{\eta} \Big )
\; .
\end{equation}
\end{lemma}

\begin{proof}
Fix $i$, $m$, and a vector $\Delta^{(i,m)} \in \calD_{i,m}$ associated with a set $C_i$ as in~\eqref{eq:Delta_i,m}.
Notice that
\[
\expectation_{j \in [N]}[|\Delta^{(i,m)}_j|]
 \leq
30 \cdot 2^{-(i+m)/2} \cdot \frac{|C_i|}{N}\; .
\]
By Corollary~\ref{cor:ChernoffCombinedNew},
applied with
\[\eps' = \tilde{\eps} \cdot 2^{(m-i)/2},~~~\alpha = \tilde{\eta},~~~{\mbox{and}}~~~a = 30 \cdot 2^{-(i+m)/2},\]
we have that~\eqref{eq:qapproxdelta} holds with probability $1-2^{-\Omega(2^m \cdot q\tilde{\eps}\tilde{\eta})}$.
Using Claim~\ref{claim:sizeHiNew}, the union bound over all
the vectors in $\calD_{i,m}$ implies that the probability
that some $\Delta^{(i,m)} \in \calD_{i,m}$ does not satisfy~\eqref{eq:qapproxdelta}
is at most
\[
N^{O (2^m \cdot \log (1/\gamma))} \cdot 2^{-\Omega (2^m \cdot q \tilde{\eps}\tilde{\eta})} \leq 2^{-\Omega(2^m \cdot \log N \cdot \log (1/\gamma))} \; .
\]
The result follows by a union bound over $i$ and $m$.
\end{proof}

\paragraph{Approximating the Vectors $Mx$.}

\begin{lemma}\label{lemma:exists_him}
For every multiset $Q \subseteq [N]$ and every vector $x \in \C^N$ with $\|x\|_1 = 1$ there exist vector collections $(\Delta^{(i,m)} \in \calD_{i,m})_{m=i,\ldots,i+r}$ associated with sets $C_i$ ($1 \leq i \leq t$), for which
\begin{enumerate}
  \item\label{itm:00} $\expectation_{j \in [N]} \bracks[\big]{|(Mx)_j|^2 } \geq \sum_{i=1}^{t}{2^{-i} \cdot \frac{|C_i|}{N}} -\eta,$
  \item\label{itm:11} $\expectation_{j \in Q} \bracks[\big]{ |(Mx)_j|^2 } \approx_{O(\eps), O(\eta)}
	       \expectation_{j \in Q} \bracks[\big]{ \sum_{i=1}^{t}{\sum_{m=i}^{i+r}{\Delta^{(i,m)}_j}} },$
        and
  \item\label{itm:22} $\expectation_{j \in [N]} \bracks[\big]{|(Mx)_j|^2 } \approx_{O(\eps), O(\eta)}
	       \expectation_{j \in [N]} \bracks[\big]{\sum_{i=1}^{t}{\sum_{m=i}^{i+r}{\Delta^{(i,m)}_j}} }.$
\end{enumerate}
\end{lemma}

\begin{proof}
By Lemma~\ref{lemma:approx}, for every $1 \leq i \leq t+r$ there exists a vector $g^{(i)} \in \calG_i$ that satisfies
\begin{align}\label{eq:Mxj_3}
|(Mx)_j| \approx_{0, 2^{-i/2}} |g^{(i)}_j|
\end{align}
for all but at most $\gamma$ fraction of $j \in [N]$ and for all but at most $\gamma$ fraction of $j \in Q$.
We say that $j \in [N]$ is {\em good} if~\eqref{eq:Mxj_3} holds for every $i$, and otherwise that it is {\em bad}. Notice that all but at most $(t+r) \gamma$ fraction of $j \in [N]$ are good and that all but at most $(t+r) \gamma$ fraction of $j \in Q$ are good. Consider the sets $C_i$ and vectors $h^{(i,m)}, \Delta^{(i,m)}$ associated with $(g^{(1)}, \ldots, g^{(t+r)})$ as defined in~\eqref{eq:H_i,m}.
We claim that $\Delta^{(i,m)}$ satisfy the requirements of the lemma.

Fix some $1 \leq i \leq t$. For every good $j \in C_i$,
the definition of $C_i$ implies that $|g^{(i)}_j| \geq 2 \cdot 2^{-i/2}$, so using~\eqref{eq:Mxj_3} it follows that
\begin{align}\label{eq:lowerMxj}
|(Mx)_j| \geq |g^{(i)}_j| -  2^{-i/2} \geq 2^{-i/2}.
\end{align}
We also claim that $|(Mx)_j| \leq 3 \cdot 2^{-(i-1)/2}$. This trivially holds for $i=1$, so assume that $i \geq 2$, and notice that the definition of $C_i$ implies that
$|g^{(i-1)}_j| < 2 \cdot 2^{-(i-1)/2}$, so using~\eqref{eq:Mxj_3}, it follows that
\begin{align}\label{eq:3_s2}
|(Mx)_j| \leq |g^{(i-1)}_j| + 2^{-(i-1)/2} \leq 3 \cdot 2^{-(i-1)/2}.
\end{align}
Since at most $(t+r) \gamma$ fraction of $j \in [N]$ are bad, \eqref{eq:lowerMxj} yields that
\begin{align*}
\expectation_{j \in [N]} \bracks[\big]{|(Mx)_j|^2}
\geq
\sum_{i=1}^{t}{2^{-i} \cdot \frac{|C_i|}{N}} -(t+r)\gamma/2 \geq \sum_{i=1}^{t}{2^{-i} \cdot \frac{|C_i|}{N}} -\eta,
\end{align*}
as required for Item~\ref{itm:00}.

Next, we claim that every good $j$ satisfies
\begin{equation}\label{eq:mxapproxh}
|(Mx)_j|^2 \approx_{O(\eps),O(\eta)} \sum_{i=1}^{t}{h^{(i,i+r)}_j} \; .
\end{equation}
For a good $j \in C_i$ and $m \geq i$,
\begin{align}\label{eq:10_i_m}
\abs[\big]{|(Mx)_j|^2 - h^{(i,m)}_j}
\leq
2 \cdot |(Mx)_j| \cdot 2^{-m/2} + 2^{-m} \leq 10 \cdot 2^{-(i+m)/2},
\end{align}
where the first inequality follows from~\eqref{eq:Mxj_3} and the second from~\eqref{eq:3_s2}.
In particular, for $m = i+r$ (recall that $r = \log_2(1/\eps^2)$), we have
\[
\abs[\big]{ |(Mx)_j|^2 - h^{(i,i+r)}_j} \leq
10 \cdot \eps \cdot 2^{-i}
\leq
10 \cdot \eps \cdot |(Mx)_j|^2
\; ,
\]
and thus $|(Mx)_j|^2 \approx_{O(\eps),0} h^{(i,i+r)}_j$. Since every good $j$ belongs to at most one of the sets $C_i$, for every good
$j \in \bigcup C_i$
we have
$|(Mx)_j|^2
\approx_{O(\eps),0} \sum_{i=1}^{t}{h^{(i,i+r)}_j}$.
On the other hand, if $j$ is good but does not belong to any $C_i$, by our choice of $t$, it satisfies
\[
|(Mx)_j| \leq
|g^{(t)}_j| + 2^{-t/2} \leq 3 \cdot 2^{-t/2} = 3 \sqrt{\eta} \; ,
\]
and thus $|(Mx)_j|^2 \approx_{0,9\eta} 0 = \sum_{i=1}^{t}{h^{(i,i+r)}_j}$.
This establishes that~\eqref{eq:mxapproxh} holds for every good $j$.

Next, we claim that for every good $j$,
\begin{align}\label{eq:eps_eta_Delta}
|(Mx)_j|^2
\approx_{O(\eps),O(\eta)}
\sum_{i=1}^{t}{\sum_{m=i}^{i+r}{\Delta^{(i,m)}_j}} \; .
\end{align}
This follows since for every $1 \leq i \leq t$, the vector $h^{(i,i+r)}$
can be written as the telescopic sum
\[
h^{(i,i+r)} = \sum_{m=i}^{i+r} \parens[\big]{h^{(i,m)}-h^{(i,m-1)}} \; ,
\]
where we used that $h^{(i,i-1)} = 0$.
We claim that for every good $j$, these differences satisfy
\[ |h^{(i,m)}_j - h^{(i,m-1)}_j | \leq 30 \cdot 2^{-(i+m)/2}, \]
thus establishing that~\eqref{eq:eps_eta_Delta} holds for every good $j$.
Indeed, for $m \geq i+1$, \eqref{eq:10_i_m} implies that
\begin{align}\label{eq:10}
|h^{(i,m)}_j - h^{(i,m-1)}_j | \leq 10 \cdot (2^{-(i+m)/2} + 2^{-(i+m-1)/2}) \leq 30 \cdot 2^{-(i+m)/2},
\end{align}
and for $m=i$ it follows from~\eqref{eq:Mxj_3} combined with~\eqref{eq:3_s2}.

Finally, for every bad $j$ we have
\[
\Big | |(Mx)_j|^2 - \sum_{i=1}^{t}{\sum_{m=i}^{i+r}{\Delta^{(i,m)}_j}} \Big |
\leq
1 + 30 \cdot \max_{1 \leq i \leq t} \Big (\sum_{m=i}^{i+r}{2^{-(i+m)/2} \Big )}
\leq
60
\; .\]
Since at most $(t+r)\gamma$ fraction of the elements in $[N]$ and in $Q$ are bad, their effect on the difference between the expectations in Items~\ref{itm:11} and~\ref{itm:22} can be bounded by $60(t+r) \gamma$. By our choice of $\gamma$ this is $\eta$, as required.
\end{proof}

Finally, we are ready to prove Theorem~\ref{thm:Main_s_}.

\begin{proof}[ of Theorem~\ref{thm:Main_s_}]
Recall that it can be assumed that $\eps \geq \eta$. By Lemma~\ref{lemma:QisGood_}, applied with $\tilde{\eps} = \eps/r$ and $\tilde{\eta} = \eta/(rt)$, a random multiset $Q$ of size
\begin{align*}
q &= O \Big (\eps^{-1} \eta^{-1} \cdot r^2 \cdot t \cdot \log N \cdot \log(1/\gamma) \Big) \\
&= O \Big ( \log^2(1/\eps) \cdot \eps^{-1} \eta^{-1} \log N \cdot \log^2 (1/\eta) \Big )
\end{align*}
satisfies with probability $1- 2^{-\Omega(\log N \cdot \log (1/\eta))}$,
that for every $1 \leq i\leq t$, $m$, and $\Delta^{(i,m)} \in \calD_{i,m}$ associated with a set $C_i$,
\begin{align*}
\expectation_{j \in Q} \bracks[\big]{\Delta^{(i,m)}_j} \approx_{0, b_i} \expectation_{j \in [N]}\bracks[\big]{\Delta^{(i,m)}_j} \mbox{ ~~for~~ }
b_i = O \Big (\frac{\eps}{r} \cdot 2^{-i} \cdot \frac{|C_i|}{N} + \frac{\eta}{rt} \Big ),
\end{align*}
in which case we also have
\begin{equation}\label{eq:deltaapproxfinal}
\expectation_{j \in Q} \bracks[\Big]{\sum_{i=1}^{t}{\sum_{m=i}^{i+r}{\Delta^{(i,m)}_j}}} \approx_{0, b} \expectation_{j \in [N]}\bracks[\Big]{\sum_{i=1}^{t}{\sum_{m=i}^{i+r}{\Delta^{(i,m)}_j}}} \mbox{ ~~for~~ }b = O \Big ( \eps \cdot \sum_{i=1}^{t}{2^{-i} \cdot \frac{|C_i|}{N}} + \eta \Big ) \; .
\end{equation}

We show that a $Q$ with the above property satisfies the requirement of the theorem. Let $x \in \C^N$ be a vector, and assume without loss of generality that $\|x\|_1 = 1$.
By Lemma~\ref{lemma:exists_him}, there exist vector collections $(\Delta^{(i,m)} \in \calD_{i,m})_{m=i,\ldots,i+r}$ associated with sets $C_i$ ($1 \leq i \leq t$), satisfying Items 1, 2, and 3 there.
Combined with~\eqref{eq:deltaapproxfinal}, this gives
\[
\expectation_{j \in Q} \bracks[\big]{|(Mx)_j|^2}
\approx_{O(\eps), O(\eta)}
\expectation_{j \in [N]}\bracks[\big]{|(Mx)_j|^2} \; ,\]
and we are done.
\end{proof}

\subsection{The Restricted Isometry Property}
It is easy to derive now the following theorem. The proof is essentially identical to that of Theorem~\ref{thm:RIP_simple}, using Theorem~\ref{thm:Main_s_} instead of Theorem~\ref{thm:Main_s}.

\begin{theorem}
For sufficiently large $N$ and $k$, a unitary matrix $M \in \C^{N \times N}$
satisfying $\|M\|_\infty \leq O(1/\sqrt{N})$, and a sufficiently small
$\eps > 0$, the following holds.
For some $q = O(\log^2(1/\eps) \eps^{-2}  \cdot k \cdot \log^2(k/\eps) \cdot \log N)$,
let $A \in \C^{q \times N}$ be a matrix whose $q$ rows are chosen
uniformly and independently from the rows of $M$,
multiplied by $\sqrt{N/q}$.
Then, with probability
$1-2^{-\Omega(\log N \cdot \log (k/\eps))}$, the matrix $A$ satisfies the restricted isometry property of order $k$ with constant $\eps$.
\end{theorem}

\bibliographystyle{abbrv}
\bibliography{rip}

\end{document}